\pdfoutput=1
\documentclass[12pt]{article}
\usepackage{pslatex,fullpage}
\usepackage[numbers]{natbib}
\usepackage[utf8]{inputenc}  
\usepackage{amsmath}
\usepackage{enumerate}
\usepackage{pslatex}
\usepackage[english]{babel}
\usepackage[final]{pdfpages}
\usepackage{colortbl}
\usepackage[ruled,vlined,linesnumbered]{algorithm2e}
\usepackage{tikz}
\usepackage{multirow}
\usepackage{subfig}
\usepackage[naturalnames]{hyperref}

\newtheorem{definition}{Definition}
\newtheorem{theorem}{Theorem}

\newtheorem{proposition}[theorem]{Proposition}

\newtheorem{proof}{Proof}

\SetKwInOut{Input}{Input}\SetKwInOut{Output}{Output}

\newcommand{\card}[1]{\vert #1 \vert}

\newcommand{\length}[1]{\vert #1 \vert}
\newcommand{\lgr}[1]{\left| #1 \right|}

\newcommand{\ov}[2]{ov(#1, #2)}

\newcommand{\norme}[1]{\vert \vert #1 \vert \vert }

\newcommand{\ovpp}{\mathcal{O}v^{+}(P)}

\newcommand{\ie}{\emph{i.e.}}

\newcommand{\mhog}{\text{MarkHOG}}
\newcommand{\bhog}{\text{bHog}}
\newcommand{\true}{\text{\texttt{True}}}
\newcommand{\false}{\text{\texttt{False}}}

\hypersetup{ 
    pdfauthor   = {Bastien Cazaux and Eric Rivals},%
    pdftitle    = {Hierarchical Overlap Graph},%
    pdfsubject  = {algorithms},%
    pdfkeywords = {overlap graph,suffix,prefix,stringology,superstring,greedy,assembly,digraph,APSP,Aho-Corasik},%
    pdfcreator  = {PDFLaTeX},%
    pdfproducer = {PDFLaTeX},
    pagebackref = true,
    colorlinks = true,
    hyperfigures = true,
    urlcolor = blue,
    citecolor = green!50!black,
    pdftex
}

\author{Bastien Cazaux, Eric Rivals\\
  L.I.R.M.M. \& Institut Biologie Computationnelle, \\
  Université Montpellier, CNRS U.M.R. 5506\\
  161 rue Ada, F-34392 Montpellier Cedex 5, France\\
  \protect\texttt{cazaux@lirmm.fr,rivals@lirmm.fr}}

\title{Hierarchical Overlap Graph}

\begin{document}
\maketitle
\begin{abstract}
  Given a set of finite words, the Overlap Graph (OG) is a complete weighted digraph where each word is a node and where the weight of an arc equals the length of the longest overlap of one word onto the other (Overlap is an asymmetric notion). The OG serves to assemble DNA fragments or to compute shortest superstrings which are a compressed representation of the input. The OG requires a space is quadratic in the number of words, which limits its scalability. The Hierarchical Overlap Graph (HOG) is an alternative graph that also encodes all maximal overlaps, but uses a space  that is linear in the sum of the lengths of the input words. We propose the first algorithm to build the HOG in linear space for words of equal length.
\end{abstract}

\section{Introduction}\label{sec:intro}

DNA assembly problem arises in bioinformatics because DNA sequencing is unable to read out complete molecules, but instead yields partial sequences of the target molecule, called reads. Hence, recovering the whole DNA sequence requires to assemble those reads, that is to merge reads according to their longest overlaps (because of the high redundancy of sequencing strategies). DNA assembly comes down to building a digraph where sequences are nodes and arcs represent possible overlaps, and second to choosing a path that dictates in which order to merge the reads. The first proposed graph was the Overlap Graph \cite{Peltola-83,TarhioU88}. Throughout this article, the input is $P:=\{s_1, \ldots,s_n\}$ a set of words. Let us denote by $\norme{P} := \sum_1^n\lgr{s_i}$.
\begin{definition}[Overlap Graph]\label{def:overlap:graph}
  The \emph{Overlap Graph} (OG) of $P$ is a complete, directed graph, weighted on its arcs, whose nodes are the words of $P$, and in which the weight of an arc $(u,v)$ equals the length of the maximal overlap from string $u$ to string $v$.
\end{definition}

In the OG (and in its variants such as the String Graph), an optimal assembly path can be computed by finding a Maximum Weighted Hamiltonian Path (which is NP-hard and difficult to approximate \cite{Blum1991}). To build the graph, one has to compute the weights of the arcs by solving the so-called All Pairs Suffix Prefix overlaps problem (APSP) on $P$. Although, Gusfield has given an optimal time algorithm for APSP in 1992, APSP has recently regained attention due to innovation in sequencing technologies that allow sequencing longer reads. Indeed, solving APSP remains difficult in practice for large datasets. Several other optimal time algorithms that improve on practical running times have been described recently, e.g.~\cite{Lim-Park-apsp-2017,Tustumi-etal-apsp-2016}.

The OG has several drawbacks. First, it is not possible to know whether two distinct arcs represent the same overlap. Second, the OG has an inherently quadratic size since it contains an arc for each possible (directed) pairs of words.
 Here, we present an alternative graph, called HOG, which represents all maximal overlaps and their relationships in term of suffix and of prefix. Since the HOG takes a space linear in cumulated lengths of the words, it can advantageously replace the OG. Note that we already gave a definition of the HOG in~\cite{Cazaux-DCC-16}. Here, we proposed the first algorithm to build the HOG in linear space.

In DNA assembly, the de Bruijn Graph (DBG) is also used, especially for large datasets of short reads. For a chosen integer $k$, reads are split into all their $k$-long substrings (termed $k$-mers), which make the nodes of the DBG, and the arcs store only $(k-1)$-long overlaps. Moreover, the relationship between reads and $k$-mers is not recorded. DBGs achieve linear space, but disregard many overlaps whose lengths differ from $k$. Hence, HOGs also represent an interesting alternative to DBGs.

\subsection{Notation and Definition of the Hierarchical Overlap Graph}
\label{sec:hog}
We consider finite, linear of strings over a finite alphabet $\Sigma$ and denote the empty string with $\epsilon$. Let $s$ be a string over $\Sigma$. We denote the length of $s$ by $\length{s}$. For any two integers $i \leq j$ in $[1,\length{s}]$, $s[i,j]$ denotes the linear substring of $s$ beginning at the position $i$ and ending at the position $j$. Then we say that $s[i,j]$ is a \emph{prefix} of $s$ iff $i=1$, a \emph{suffix} iff $j=\lgr{s}$. A prefix (or suffix) $s'$ of $s$ is said \emph{proper} if $s'$  differs from $s$. For another linear string $t$, an overlap from $s$ to $t$ is a proper suffix of $s$ that is also a proper prefix of $t$. We denote the longest such overlaps by $\ov{s}{t}$.
For $A,B$ any two boolean arrays of the same size, we denote by $A \wedge B$ the boolean operation \texttt{and} between $A$ and $B$.

\begin{figure}[htbp]
  \begin{minipage}[b]{0.375\linewidth}
    \subfloat[][Aho Corasik tree of $P$]{\includegraphics[scale=.2703675]{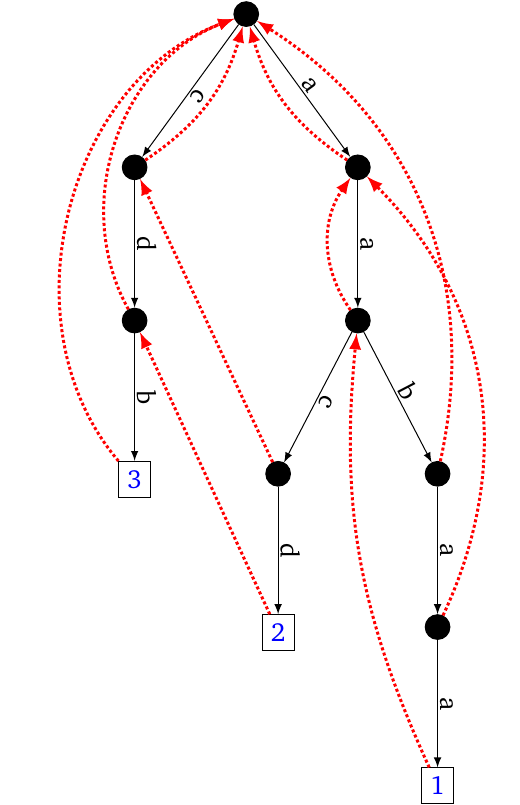}\label{toyaho}}    
  \end{minipage}
  \begin{minipage}[b]{0.3\linewidth}
    \subfloat[][EHOG of $P$]{\includegraphics[scale=.34675]{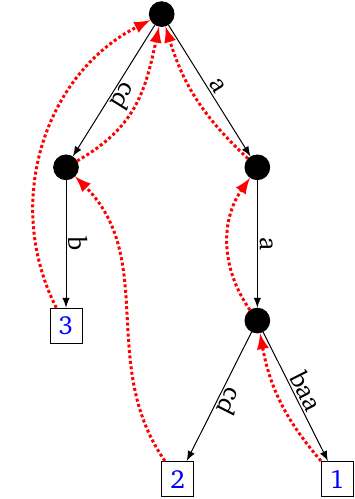}\label{toyehog}}    
  \end{minipage}
  \begin{minipage}[b]{0.3\linewidth}
    \subfloat[][HOG of $P$]{\includegraphics[scale=.34675]{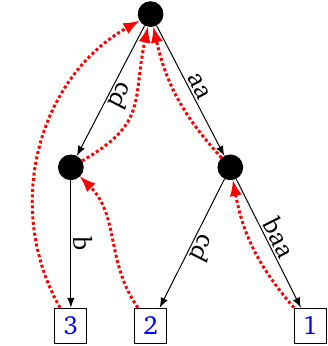}\label{toyhog}}    
  \end{minipage}
  \caption{\label{fig:aho} Consider instance $P := \{ aabaa, aacd, cdb \}$. \protect\subref{toyaho} Aho Corasik tree of $P$ with \texttt{Failure Links} in dotted lines.
     \protect\subref{toyehog} EHOG of $P$ - compared to \protect\subref{toyaho} all nodes that were not overlaps were removed, and the arcs contracted. \protect\subref{toyhog} HOG of $P$; compared to \protect\subref{toyehog} the node $a$ has been removed, and some arcs contracted.}
\end{figure}

\noindent
Let $\ovpp$ be the set of all overlaps between words of $P$. 
Let $\mathcal{O}v(P)$ be the set of \textbf{maximum} overlaps from a string of $P$ to another string or the same string of $P$.

Let us define the \emph{Extended Hierarchical Overlap Graph} and the \emph{Hierarchical Overlap Graph} as follows.

\begin{definition}\label{def:ehog}
  The \emph{Extended Hierarchical Overlap Graph} of $P$, denoted by $EHOG(P)$, is the directed graph $(V^{+}, E^{+})$  where $V^{+} = P\cup \ovpp$ and
  $E^{+}$ is the set:
    \[\{(x,y) \in (P\cup \ovpp)^2 \mid y \text{ is the longest proper suffix of } x \text{ or } x \text{ is the longest proper prefix of } y \}.\]
  \label{def:hog}
  The \emph{Hierarchical Overlap Graph} of $P$, denoted by $HOG(P)$, is the directed graph $(V, E)$ where $V := P\cup \mathcal{O}v(P)$ and $E$ is the set:
    \[\{(x,y) \in (P\cup \mathcal{O}v(P))^2 \mid y \text{ is the longest proper suffix of } x \text{ or } x \text{ is the longest proper prefix of } y \}.\]
\end{definition}
Remark that $\mathcal{O}v(P)$ is a subset of $\ovpp$. Both definitions are identical except that each occurrence of $\ovpp$ in the $EHOG(P)$ is replaced by $\mathcal{O}v(P)$ in the $HOG(P)$. Hence, $HOG(P)$ is somehow included in $EHOG(P)$. Examples of EHOG and HOG are shown in Figures~\ref{toyehog}, \ref{toyhog}.

\subsection{Related works}
\label{sec:related}
In 1990, Ukkonen gave a linear time implementation of the greedy algorithm for finding a linear superstring of a set $P$~\cite{Ukkonen-linear-greedy-1990}. This algorithm is based on the Aho-Corasick (AC) automaton \cite{Aho-Corasick-75}. In AC automaton, the input words are spelled out along a tree structure and each node represents a unique prefix of these words ; hence, the leaves represent the input words (Example in Figure~\ref{toyaho}: the tree structure, which are  is in black).  The AC automaton contains additional arcs called, \texttt{Failure links}, which link two nodes say $(x,y)$ if $y$ is the longest suffix of $x$ that is a node in the tree. 
Ukkonen characterised which nodes of the tree correspond to overlaps between words of $P$~\cite[Lemma~3]{Ukkonen-linear-greedy-1990}. In the EHOG, beyond the leaves, we keep only those nodes that are overlaps (Figure~\ref{toyehog}). From this lemma, it follows that $EHOG(P)$ is embedded in the Aho-Corasick automaton: the nodes of $EHOG(P)$ are a subset of those of the automaton, its arcs are contracted \texttt{goto} transitions or contracted \texttt{Failure Links}.

\textbf{Algorithm to build the EHOG}. Given the Aho-Corasick automaton of $P$ (whose underlying structure is a tree spelling out the words of $P$ from the root -- this tree is commonly called the \emph{trie}), for each leaf, follow the chain of \texttt{Failure links} up to the root and mark all visited nodes. Thanks to \cite[Lemma~3]{Ukkonen-linear-greedy-1990} all such nodes are overlaps. Another traversal of the AC automaton suffices to remove unmarked nodes and to contract the appropriate arcs. This algorithm takes $O(\norme{P})$ time. An involved algorithm to build a memory compact version of the EHOG is detailed in~\cite{CanovasCR17}. In the sequel, whenever we refer to the \emph{tree} structure, it is the tree structure of the EHOG, we mean the tree defined by the \texttt{goto} arcs in the EHOG (which appears in black in Figure~\ref{toyehog}).

An algorithm for computing a shortest cyclic cover of a set of DNA sequences needs to build either the EHOG or HOG~\cite{Cazaux-DCC-16}. There, we stated mistakingly in Theorem~3 that the HOG could be built in $O(\norme{P})$ time, although we meant the EHOG. The construction algorithm uses the Generalised Suffix Tree of $P$ to detect nodes representing overlaps as explained in~\cite{Gusfield-apsp-1992}.

The HOG is reminiscent of the "Hierarchical Graph" of~\cite{Golovnev-exact-IPL-2014}, where the arcs also encodes inclusion between substrings of the input words. However, in the Hierarchical Graph, each arc extends or shortens the string by a single symbol, making it larger than the HOG.


\section{Construction algorithm for the HOG}
\label{sec:algo:hog}

All internal nodes of $EHOG(P)$ are overlaps between words of $P$, while those of $HOG(P)$ are \textit{maximal} overlaps between words of $P$.  Given the EHOG of $P$, to build $HOG(P)$ we need to discard nodes that are not maximal overlaps and to merge the arcs entering and going out of such nodes. This processing can be performed in linear time on the size of $EHOG(P)$ provided that the nodes of $HOG(P)$ are known. We present Algorithm~\ref{algo:hog}, which recapitulates this construction procedure. Once $EHOG(P)$ is built, each of its internal node $u$ is equipped with a list $R_l(u)$, whose meaning is explained below. Then, Algorithm~\ref{algo:hog} calls the key procedure $\mhog(r)$ on line~\ref{line:call:mhog} to mark the nodes of $HOG(P)$ in a global boolean array denoted $\bhog$, which we store in a \emph{bit vector}. The procedure for $\mhog$ is given in Algorithm~\ref{algo:mhog}. Once done, it contracts $EHOG(P)$ in place to obtain $HOG(P)$.

\begin{algorithm}[htb]
  \DontPrintSemicolon
  \Input{$P$ a substring free set of words}
  \Output{$HOG(P)$; \textbf{Variable}: $\bhog$ a bit vector of size $\card{EHOG(P)}$} 
  build $EHOG(P)$\;
  set all values of $\bhog$ to $\false$\;
  traverse $EHOG(P)$ to build $R_l(u)$ for each internal node $u$ of $EHOG(P)$\;
  run $\mhog(r)$ where  $r$ is the root of $EHOG(P)$\label{line:call:mhog}\;%
  {Contract}$(EHOG(P), \bhog)$\label{line:hog:contract}\; 
  \tcp*{Procedure $\text{Contract}$ traverses $EHOG(P)$ to discard nodes that are not marked in $\bhog$ and contract the appropriate arcs}
  \caption{$HOG$ construction\protect\label{algo:hog}}
\end{algorithm}

\paragraph{Meaning of $R_l(u)$}%
For any internal node $u$, $R_l(u)$ lists the words of $P$ that admit $u$ as a suffix. Formally stated:
$R_l(u) := \{i \in \{1,\ldots,\card{P}\} : u \text{ is suffix of } s_i\}$.
As we need to iterate over $R_l(u)$, it is convenient to store it in a list of integers.
A traversal of $EHOG(P)$ allows to build a list $R_l(u)$ for each internal node $u$ as stated in~\cite{Ukkonen-linear-greedy-1990}. Remark that, while our algorithm processes $EHOG(P)$, Ukkonen's algorithm processes the full trie or Aho Corasik automaton, in which $EHOG(P)$ is embedded. The cumulated sizes of all $R_l$ is linear in $\norme{P}$ (indeed, internal nodes represent different prefixes of words of $P$ and have thus different begin/end positions in those words).

\paragraph{Node of the EHOG and overlaps}
The following proposition states an important property of EHOG nodes in terms of overlaps. It will allow us not to consider all possible pairs in $P\times P$, when determining maximum overlaps.

\begin{proposition}\label{prop:node:ov}
  Let $u$ be an internal node of $EHOG(P)$ and let $i \in R_l(u)$. Then, for any word $s_j$ whose leaf is in the subtree of $u$, we have
  $\ov{s_i}{s_j} \geq \lgr{u}$.
\end{proposition}
\begin{proof}
  Indeed, as $u$ belongs to $EHOG(P)$, we get that $u$ is  an overlap from $s_i$ to $s_j$.
  Thus, their maximal overlap has length at least $\lgr{u}$. \hfill \textbf{\textsc{QED}}
\end{proof}

\paragraph{Description of Algorithm~\ref{algo:mhog}}
We propose Algorithm~\ref{algo:mhog}, a recursive algorithm that determines which nodes belong to $\mathcal{O}v(P)$ while traversing $EHOG(P)$ in a depth first manner, and marks them in $\bhog$. At the end of the algorithm, for any node $w$ of $EHOG(P)$, the entry $\bhog[w]$ is \true\ if and only if $w$ belongs to $HOG(P)$, and \false\ otherwise.

By the definition of a HOG, the leaves of the $EHOG(P)$, which represent the words of $P$, also belong to $HOG(P)$ (hence, line~\ref{line-bhog-leaf}).

The goal is to distinguish among internal nodes those representing maximal overlaps (\ie, nodes of $\mathcal{O}v(P)$) of at least one pair of words. We process internal nodes in order of decreasing word depth.
By Proposition~\ref{prop:node:ov}, we know that for any $i$ in $R_l(u)$, $s_i$ overlaps any $s_j$ whose leaf is in the subtree of $u$. However, $u$ may not be the maximal overlap for some pairs, since longer overlaps may have been found in each subtree. Indeed, $u$ can be a maximal overlap from $s_i$ onto some $s_j$, if and only if for any child $v$ of $u$, $s_i$ has not already a maximal overlap with the leaves in the child's subtree.  Hence, to check this, we compute the $C$ vector for each child by recursively calling $\mhog$ for each child, and merge them with a boolean \texttt{and} (line~\ref{line:mhog:and}). We get invariant~\#1:\\
\centerline{$C[w]$ is \true\ iff for any leaf $l$ in the subtree of $u$ the pair $\ov{w}{l} \mathbf{>} \lgr{u}$.}
Then, we scan $R_l(u)$, for each word $w$ if $C[w]$ is \false, then for at least one word $s_j$, the maximum overlap of $w$ onto $s_j$ has not been found in the subtree.  By Proposition~\ref{prop:node:ov}, $u$ is an overlap from $w$ onto $s_j$, and thus we set both $C[w]$ and $\bhog[u]$ to \true\ (lines~\ref{algo:mhog:update:deb}-\ref{algo:mhog:update:fin}). Otherwise, if $C[w]$ is \true, then it remains so, but then $\bhog[u]$ remains unchanged.  This yields Invariant \#2:\\
\centerline{$C[w]$ is \true\ iff for any leaf $l$ in the subtree of $u$ the pair $\ov{w}{l} \mathbf{\geq} \lgr{u}$.}
This ensures the correctness of $\mhog$.

\begin{algorithm}[htb]
  \DontPrintSemicolon
  \Input{$u$ a node of $EHOG(P)$}
  \Output{$C$: a boolean array of size $\card{P}$}
  %
  \If{$u$ is a leaf}{
    set all values of $C$ to $\false$\;
    $\bhog[u] := \true$\label{line-bhog-leaf}\;
    \Return $C$\;
  }
  %
  \tcp{Cumulate the information for all children  of $u$}
  $C := \mhog(v)$ where $v$ is the first child of $u$\;
  \ForEach{$v$ among the other children of $u$}{
    $C := C \wedge \mhog(v)$\protect\label{line:mhog:and}\;
  }
  \tcp{Invariant \textbf{1}: $C[w]$ is \true\ iff for any leaf $s_j$ in the subtree of $u$ the pair $\ov{w}{s_j}> \lgr{u}$}
  %
  \tcp{Process overlaps arising at node $u$: Traverse the list $R_l(u)$}
  \For{node $x$ in the list  $R_l(u)$} {
    \If{$C[x] = \false$     \label{algo:mhog:update:deb}}	
    {$\bhog[u] := \true$}
    $C[x] := \true$\;    \label{algo:mhog:update:fin}
  }
  \tcp{Invariant \textbf{2}: $C[w]$ is \true\ iff for any leaf $s_j$ in the subtree of $u$ the pair $\ov{w}{s_j} \geq \lgr{u}$}
  \Return $C$
  \caption{$\mhog(u)$; \label{algo:mhog}}
\end{algorithm}


The complexity of Algorithm~\ref{algo:hog} is dominated by the call of the recursive algorithm $\mhog(r)$ in line~\ref{line:call:mhog}, since computing $EHOG(P)$, building the lists $R_l(.)$ and contracting the EHOG into the HOG
 (line~\ref{line:hog:contract}) all take linear time in $\norme{P}$.

How many simultaneous calls of $\mhog$ can there be? Each call uses a bit vector $C$ of length $\card{P}$, which impacts the space complexity. The following proposition helps answering this question.
\begin{proposition}\label{prop:mhog:calls}
  Let $u,v$ be two nodes of  $EHOG(P)$. Then if $u$ and $v$ belong to distinct subtrees, then $\mhog(u)$ terminates before $\mhog(v)$ begins, or vice versa.
\end{proposition}

Hence, the maximum number of simultaneously running $\mhog$ procedures is bounded by the maximum (node) depth of $EHOG(P)$, which is itself bounded by the length of the longest word in $P$.
Now, consider the amortised time complexity of Algorithm~\ref{algo:mhog} over all calls. For each node, the corresponding $C$ vector is computed once and merged once during the processing of his parent.
Moreover, if amortised over all calls of $\mhog$, the processing all the lists $R_l$ for all nodes take linear time in $\norme{P}$. Altogether all calls of Algorithm~\ref{algo:mhog} require  $O(\norme{P} + \card{P}^2)$ time.
The space complexity sums the space occupied by $EHOG(P)$, \ie\ $O(\norme{P})$, and that of the $C$ arrays.
Now, the maximum number of simultaneously running calls of $\mhog$ is the maximum number of simultaneously stored $C$ vectors. This number is bounded by the maximum number of branching nodes on a path from the root to a leaf of $EHOG(P)$, which is smaller than the minimum among $\max\{\lgr{s}: s \in P\}$ and $\card{P}$. Hence, the space complexity is $O(\norme{P}+ \card{P}\times\min(\card{P}, \max\{\lgr{s}: s \in P\})$.

\begin{theorem}\label{th:hog:complex}
  Let $P$ be a set of words. Then 
  Algorithm~\ref{algo:hog} computes $HOG(P)$ using  $O(\norme{P} + \card{P}^2)$ time and $O(\norme{P}+ \card{P}\times\min(\card{P}, \max\{\lgr{s}: s \in P\})$ space.
If all words of $P$ have the same length, then the space complexity is $O(\norme{P})$.
\end{theorem}

\paragraph{An example of HOG construction}
Consider the instance $P := \{ tattatt, ctattat, gtattat, cctat \}$; the graph $EHOG(P)$ is shown Figure~\ref{a}.  Table~\ref{b} traces the execution of Algorithm~\ref{algo:mhog} for instance $P$. It describes for each internal node (leaves not included), the list $R_l$, the computation of $C$ vector, the final value of $\bhog$ and the list of specific pairs of words for which the corresponding node is a maximum overlap. We write $(x/y, z)$ as a shortcut for pairs $(x,z)$ and $(y,z)$. Another example is given in appendix.

\begin{figure}[htbp]
  \begin{minipage}{0.3725\linewidth}
    \hspace{-0.5cm}
    \subfloat[][EHOG of $P$]{\includegraphics[scale=.405675]{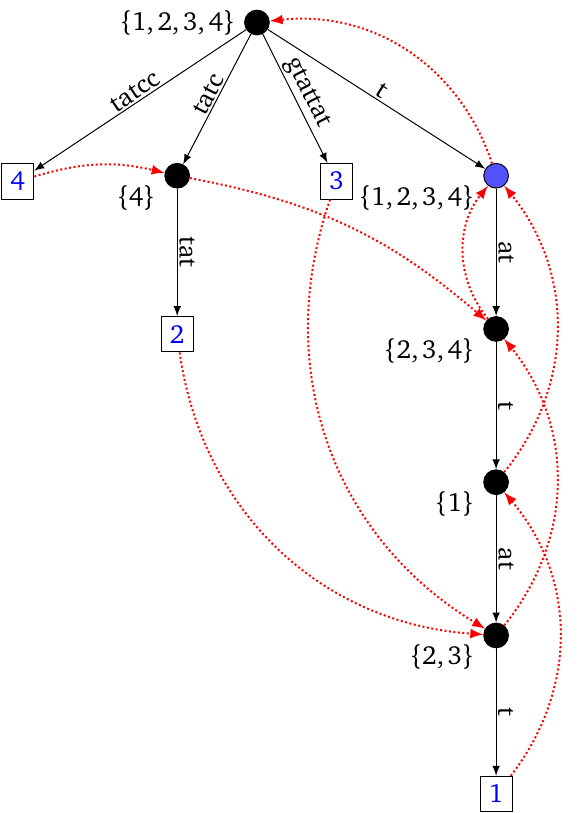}\label{a}}    
  \end{minipage}
  \begin{minipage}{0.62745\linewidth}
    \centering
    \subfloat[t][Weights of OG of $P$ in a matrix]{\label{b}
      \centering
      \setlength{\tabcolsep}{2\tabcolsep} 
      \begin{tabular}[b]{>{\columncolor[gray]{.8}}ccccc}
        \rowcolor[gray]{.6}
        (x,y) & 1 & 2 & 3 & 4\\
        1     & 4 & 0 & 0 & 0\\
        2     & 6 & 0 & 0 & 0\\
        3     & 6 & 0 & 0 & 0\\
        4     & 3 & 4 & 0 & 0
      \end{tabular}
      \setlength{\tabcolsep}{.5\tabcolsep}
    }
    \\
    \subfloat[][]{\small\label{c}
      \begin{tabular}[b]{>{\columncolor[gray]{.8}}llrrlr}
        \rowcolor[gray]{.6}
        node &$R_{\ell}$ & $C$(before) & $C$(after) &Specific pairs &$\bhog$\\
        ctat & \{4\} & 0000 & 0001 & (4,2) & 1\\
        tattat & \{2,3\} & 0000 & 0110 & (2,1) (3,1) & 1\\
        tatt & \{1\} & 0110 & 1110 & (1,1) & 1\\
        tat & \{2,3,4\} & 1110 & 1111 & (4,1) & 1\\
        t & \{1,2,3,4\} & 1111 & 1111 & empty & 0\\
        root & \{1,2,3,4\} & 0000 \^{} 0001 & 0000 &  &\\
        root & \{1,2,3,4\} & 0000 \^{} 0000 & 0000 &(2/3,2)   &\\
        root & \{1,2,3,4\} & 0000 \^{} 1111 & 0000 & (1/2/3/4,4)  &\\
        root & \{1,2,3,4\} & 0000 & 1111 & (2/3/4,3) & 1
      \end{tabular}
    }
  \end{minipage}
  \caption{\label{fig:ehog} \protect\subref{a} EHOG for instance $P := \{ tattatt, ctattat, gtattat, cctat \}$. \texttt{goto} transitions appear in black arcs, \texttt{Failure Links} in dotted red arcs. For each internal node, the list $R_L$ is given between brackets. \protect\subref{b} Overlap Graph of $P$ given in matrix form. \protect\subref{c} Trace of Algorithm~\ref{algo:mhog}.  For each internal node are shown: the word it represents, $R_{l}$, the bit vector $C$ when before and after the node is processed, \bhog, and the pairs for which it is a maximum overlap.  The node $t$ is the only internal node which  is not a maximal overlap for some pair, and indeed \bhog\ is set to 0. The computation for the root node shows each \texttt{and} between $C$ vectors from the children on four lines.}
\end{figure}



\section{Conclusion}\label{sec:conclusion}

The Hierarchical Overlap Graph (HOG) is a compact alternative to the Overlap Graph (OG) since it also encodes all maximal overlaps as the OG, but in a space that is linear in the norm of $P$. In addition, the HOG records the suffix- and the prefix-relationship between the overlaps, while the OG lacks this information, which is useful for computing greedy superstrings~\cite{Ukkonen-linear-greedy-1990,Cazaux-Rivals-JDA-16}. Because the norm of $P$ can be large in practice, it is thus important to build the HOG also in linear space, which our algorithm achieves if all words have the same length. 

For constructing the HOG, Algorithm~\ref{algo:hog} takes $O(\norme{P} + \card{P}^2)$ time.  Whether one can compute the HOG in a time linear in $\norme{P} + \card{P}$ remains open. An argument against this complexity is that all the information needed to build the Overlap Graph is encoded in the HOG.

The EHOG and HOG differs by definition, and the nodes of the HOG are a subset of the nodes of the EHOG. In practice or in average, is this difference in number of nodes substantial? There exist instances such that in the limit the ratio between the number of nodes of EHOG versus HOG tends to infinity when $\norme{P}$ tends to infinity while the number of input words remains bounded (see appendix). For instance, a small alphabet (e.g. DNA) favors multiple overlaps and tend to increase the number of nodes that are proper to the EHOG, while for natural languages HOG and EHOG tend to be equal.

For some applications like DNA assembly, it is valuable to compute approximate rather than exact overlaps. The approach proposed here does not easily extend to approximate overlaps. Some algorithms have been proposed to compute OG with arcs representing approximate overlaps, where approximation is measured by the Hamming or the edit distance \cite{ValimakiLM12}.

\paragraph{Acknowledgements:} {This work is supported by the Institut de Biologie Computationnelle ANR (ANR-11-BINF-0002), and D\'efi {MASTODONS}  at \href{http://www.lirmm.fr/~rivals/C3G}{C3G}  from CNRS.

{
  \bibliographystyle{plain}
  \bibliography{hog-art-arxiv}
}
\pagebreak
\appendix
\section*{Appendix A: Difference in number of nodes between EHOG and HOG}
\label{sec:diff:ehog:hog}

Consider a finite alphabet, say $\Sigma = \{a,c,g,t\}$.
Let $s$ be any word formed by a permutation of $\Sigma$. Here, $s$ could be $s := acgt$. The length of $s$ is the cardinality of $\Sigma$.

Let $z$ be a positive integer. We build the following  instance $P_z$ by taking :
\begin{itemize}
\item the word $w$ made of the concatenation of $z$ copies of word $s$,
\item and $(\card{\Sigma}-1)$ other words, which are $(\card{\Sigma}-1)$ cyclic shifts of $w$, denoted $w_1, \ldots, w_{(\card{\Sigma}-1)}$.
\end{itemize}
For a letter $\alpha$ and a word $v$, the (first) cyclic shift of the word ${\alpha}v$ is $v{\alpha}$. In our construction,
\begin{itemize}
\item $w_1$ is the cyclic shift of $w$, 
\item $w_2$ is the cyclic shift of $w_1$,
  \\[.1cm]
  \ldots
\item $w$ is the cyclic shift of $w_{(\card{\Sigma}-1)}$.
\end{itemize}
Note that  $\card{P_z} =  \card{\Sigma} = \lgr{s}$ and that 
$\norme{P_z} =  \card{P_z} \times \lgr{w} =  \card{P_z} \times z \times \lgr{s}  =  \card{P_z}^2 \times z$

For an  instance $P_z$,
\begin{itemize}
\item the $EHOG(P_z)$ contains exactly as many nodes as the Aho-Corasick automaton contains states, that is $\norme{P_z}$ nodes (leaves included). Indeed, all internal nodes are overlaps.
\item The $HOG(P_z)$ contains $\card{P_z}^2$ internal nodes and $\card{P_z}$ leaves.
\end{itemize}

Let us denote the number of nodes of the EHOG and of the HOG respectively by $\lgr{EHOG(P_z)}$ and  $\lgr{HOG(P_z)}$.
Then the ratio
\begin{eqnarray}\label{eq:ratio}
  \frac{\lgr{EHOG(P_z)}}{\lgr{HOG(P_z)}} &= & \frac{\norme{P_z}}{\card{P_z}^2 + \card{P_z}}\\
                                          &= & \frac{z \times \card{P_z}^2}{\card{P_z}^2 + \card{P_z}}\\
                                          &= & \frac{z}{1+ \frac{1}{\card{P_z}}}\\
                                          &\to_{z \rightarrow \infty} &+ \infty.
\end{eqnarray}
\section*{Appendix B: An example of HOG construction}
\label{sec:hog:ex:deux}
Consider the instance $P := \{ bcbcb, baba, abcba, abab \}$.
Its OG and EHOG are given in the figure below.
The HOG is exactly the EHOG, since all internal nodes of the EHOG are in fact maximal overlaps for some pairs of words of $P$.
\begin{figure}[h!]
  \centering
  \subfloat[][Overlap Graph (OG) of $P$]{\includegraphics[scale=.4675]{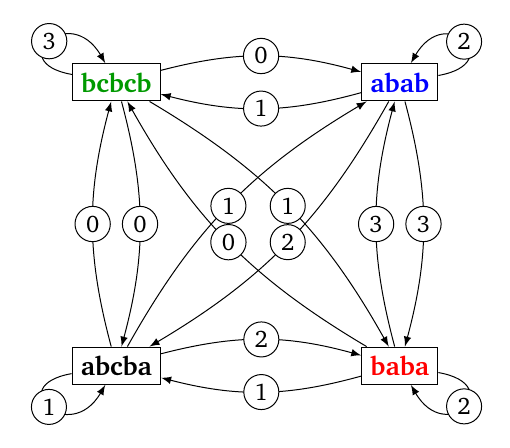}\label{sa}}
  \subfloat[][Both $EHOG(P)$ and $HOG(P)$]{\includegraphics[scale=.4675]{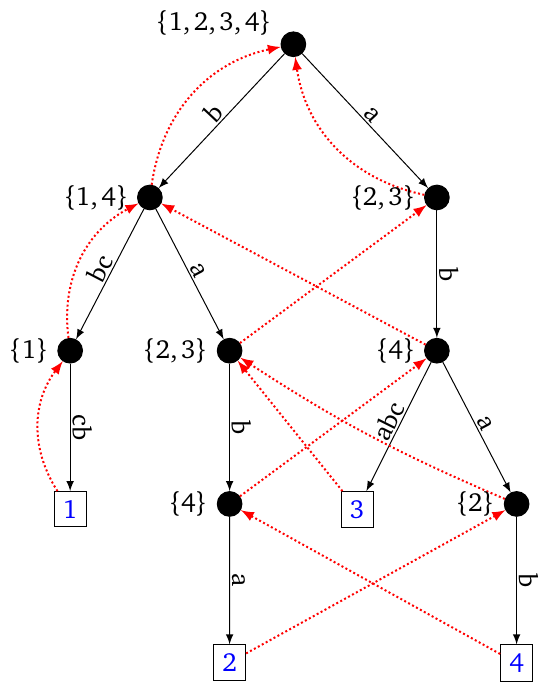}\label{sb}}
  \\
  \subfloat[][]{\small\label{sc}
    \begin{tabular}[b]{>{\columncolor[gray]{.8}}llrrlr}
      \rowcolor[gray]{.6}
      node &$R_{\ell}$ & $C$(before) & $C$(after) &Specific pairs &$\bhog$\\
      bcb & \{1\} & 0000 & 1000 & (1,1) & 1\\
      bab & \{4\} & 0000 & 0001 & (4,2) & 1\\
      ba & \{2,3\} & 0001 & 0111 & (2,2) (3,2) & 1\\
      b & \{1, 4\} & 1000 \^{} 0111 &  &  & \\
      b & \{1, 4\} & 0000 & 1001 & (4,1) (1,2) & 1\\
      aba & \{2\} & 0000 & 0100 & (2,4) & 1\\
      ab & \{4\} & 0000 \^{} 0100 &  &  & \\
      ab & \{4\} & 0000 & 0001 & (4,3) (4,4) & 1\\
      a & \{2,3\} & 0001 & 0111 & (2,3) (3,3) (3,4) & 1\\
      root & \{1,2,3,4\} & 1001 \^{} 0111 &  &  & \\
      root & \{1,2,3,4\} & 0001 & 1111 & (1,3) (1,4) (2,1) (3,1) & 1\\
    \end{tabular}
  }
  \caption{\label{fig:exsup} \protect\subref{sa} and \protect\subref{sb} OG and EHOG for instance $P := \{ bcbcb, baba, abcba, abab \}$. In the EHOG, \texttt{goto} transitions appear in black arcs, \texttt{Failure Links} in dotted red arcs. For each internal node, the list $R_L$ is given between bracket. \protect\subref{sc} Trace of Algorithm~\mhog.  For each internal node are shown: the word it represents, $R_{l}$, the bit vector $C$ when before and after the node is processed, \bhog, and the pairs for which it is a maximum overlap.  All nodes of the EHOG also belong to the HOG (as shown by the values of \bhog\ being 1).}
\end{figure}
\end{document}